\documentclass[10pt,twocolumn,letterpaper]{article}

\usepackage{cvpr}
\usepackage{times}
\usepackage{epsfig}
\usepackage{graphicx}
\usepackage{amsmath}
\usepackage{amssymb}

\usepackage{epigraph}
\usepackage{times}
\usepackage{epsfig}
\usepackage{graphicx}
\usepackage{amsmath}
\usepackage{amssymb}
\usepackage{adjustbox}
\usepackage{amsthm}
\usepackage{bm}
\usepackage[ruled]{algorithm2e} % For algorithms

\newtheorem{prop}{Proposition}[section]
\newtheorem{theorem}{Theorem}

\DeclareMathOperator*{\argmax}{arg\,max}
\newcommand{\qyl}[1]{{\color{black}#1}}

\newcommand{\YL}[1]{{\color{black}#1}}

\usepackage[breaklinks=true,bookmarks=false]{hyperref}

\cvprfinalcopy % *** Uncomment this line for the final submission

\def\cvprPaperID{****} % *** Enter the CVPR Paper ID here

\begin{document}

%%%%%%%%% TITLE
\title{Learning-based Real-time Detection of Intrinsic Reflectional Symmetry}

\author{Yi-Ling Qiao\textsuperscript{1,3}, Lin Gao\textsuperscript{1}\thanks{Corresponding Author}, Shu-Zhi Liu\textsuperscript{1,2}, Ligang Liu\textsuperscript{4}, Yu-Kun Lai \textsuperscript{5}, Xilin Chen\textsuperscript{1}\\
\textsuperscript{1}Institute of Computing Technology, Chinese Academy of Sciences \\
\textsuperscript{2}University of Chinese Academy of Sciences\\
\textsuperscript{3}University of Maryland, College Park\\
\textsuperscript{4}University of Science and Technology of China\\
\textsuperscript{5}Cardiff University\\
qiaoyiling15@mails.ucas.ac.cn, gaolin@ict.ac.cn, liushuzhi16@mails.ucas.ac.cn \\ lgliu@ustc.edu.cn, LaiY4@cardiff.ac.uk, xlchen@ict.ac.cn
}

\maketitle
%\thispagestyle{empty}

%%%%%%%%% ABSTRACT
\begin{abstract}
Reflectional symmetry is ubiquitous in nature. While extrinsic reflectional symmetry can be easily parametrized and detected, intrinsic symmetry is much harder due to the high solution space. Previous works usually solve this problem by voting or sampling, which suffer from high computational cost and randomness. In this paper, we propose \YL{a} learning-based approach to  intrinsic reflectional symmetry detection. Instead of directly finding symmetric point pairs, we parametrize this self-isometry using a functional map matrix, which can be easily computed  given the signs of Laplacian eigenfunctions under the symmetric mapping. Therefore, we train a novel deep neural network to predict the sign of each eigenfunction under symmetry, which in addition takes the first few eigenfunctions as intrinsic features to characterize the mesh while avoiding coping with the connectivity explicitly. Our network aims at learning the global property of functions, and consequently converts the problem defined on the manifold to the functional domain. By disentangling the prediction of the matrix into separated basis, our method generalizes well to new shapes and is invariant under perturbation of eigenfunctions. Through extensive experiments, we demonstrate the robustness of our method in challenging cases, including different topology and incomplete shapes with holes. By avoiding random sampling, our learning-based algorithm is over 100 times faster than state-of-the-art methods, and meanwhile, is more robust, achieving higher correspondence accuracy in commonly used metrics.
\end{abstract}

%%%%%%%%% BODY TEXT

\section{Introduction}
\epigraph{What is it indeed that gives us the feeling of elegance in a solution, in a demonstration? It is the harmony of the diverse parts, their symmetry, their happy balance.}{\textit{Henri Poincar\'e}}

Symmetry is a common pattern that appears ubiquitously in the world. Majority of living things, including humans, animals, and plants (e.g. flowers) have some form of symmetry. 
It is also a widely employed design principle in man-made objects, including buildings, furniture, vehicles, to name a few. 

Due to its wide applicability, symmetry information has been exploited in computer graphics to help address many geometry processing tasks, including shape matching~\cite{tevs2014relating}, segmentation~\cite{dessein2017symmetry}, editing~\cite{mitra2014structure}, completion~\cite{speciale2016symmetry}, and understanding~\cite{mitra2010illustrating}.
% need more citation on related work
\YL{In these application systems, symmetry detection is usually an integral component, so efficient symmetry detection, especially achieving real-time efficiency, has significant benefits, e.g., to avoid impeding real-time performance in 3D acquisition/reconstruction, and for improved user experience in interactive shape editing by reducing users' waiting time.}

In geometry processing, researchers mainly focus on symmetry in the spatial domain, including extrinsic symmetry defined in Euclidean space or intrinsic symmetry defined in non-Euclidean (manifold) space.
Extrinsic symmetry refers to shape invariance w.r.t. rigid (including reflectional)  transformations. 
Compared to extrinsic symmetry, intrinsic symmetry is more difficult to detect due to its much larger solution space, \qyl{as discussed in previous work~\cite{wang2017group,kim2010mobius,ovsjanikov2008global}}.
 
Given a shape model, intrinsic symmetry detection aims to estimate a self-homeomorphism on the manifold that preserves the geodesic distance between each point pair. Usually, the manifold is discretized as a triangle mesh, and algorithms predict a point-wise correspondence matrix to represent symmetric pairs. \YL{State-of-the-art methods for intrinsic symmetry detection is largely based on embedding the symmetry space to some lower dimensional spaces, such as M{\"o}bius transformation space~\cite{kim2010mobius}, Global Point Signature (GPS) space~\cite{ovsjanikov2008global}, or functional map space~\cite{wang2017group,nagar2018fast} and performing random sampling or voting, which suffers from high computational cost and uncertainty of results due to randomness.
}

Despite great effort, efficient and robust detection of intrinsic symmetry remains challenging. Existing state-of-the-art methods typically take several seconds or longer to analyze one shape~\cite{nagar2018fast}, and may produce unreliable results for difficult cases.
To address this, we design the first learning-based intrinsic symmetry detection method 
%to learn such an elegant pattern across different objects. 
%we propose a learning-based method 
to handle the intrinsic symmetry problem.
%We use a functional map to represent the symmetry, which describes the self-mapping in the functional space.
Like most existing works, we focus on intrinsic reflectional symmetry as it is  most common in the real world.
Learning intrinsic symmetry directly on meshes is challenging, due to their irregular connectivity, and the global nature of symmetry.
We simplify this problem when designing the deep neural network, such that it does not directly process the edges and faces of the mesh, but instead takes intrinsic features as input.

%Since the rotational symmetry is notoriously difficult to detect due to its high dimension and the eigenfunctions of rotational symmetric objects are complex, we focus on reflectional symmmetry in this paper.* 
Similar to~\cite{nagar2018fast}, given an input mesh, the symmetry mapping defined on it can be represented using a functional map, or equivalently using a functional map matrix.  Laplace-Beltrami eigenfunctions can be extracted to provide a basis for analysis. In the matrix, entries corresponding to eigenfunctions associated with non-repeating eigenvalues are determined by the sign (odd or even) of the eigenfunction after the symmetry mapping is applied. State-of-the-art work~\cite{nagar2018fast} determines the sign of the eigenfunction through random sampling. Although it is faster than previous methods, it is still slow (requiring several seconds for a typical mesh), and may not be sufficiently robust. 

To address this, we train a deep neural network to predict the sign of each eigenfunction.
%, .e., whether it would change the sign (odd) or not (even) after symmetry mapping.
We design SignNet, a deep neural network for sign prediction, that in addition to the eigenfunction to be predicted, also takes the first few Laplacian eigenfunctions as input, which effectively encode intrinsic descriptions of the mesh characteristics, while avoiding coping with mesh connectivity explicitly. 
%%%YKL some motivation is needed
To make the computation more efficient, we truncate Laplace-Beltrami eigenfunctions in the spectral domain to lower the dimension of the representation. 
After predicting the entries of the functional map matrix, we apply a post-processing to further fine-tune the results (addressing issues such as near-identical eigenvalues and slight non-isometry) and convert the functional map to one-to-one point correspondence.

The main contributions of this work are summarized as follows:
\begin{itemize}
% think more about the contribution
\item \YL{We propose the first learning-based method to detect global intrinsic reflectional symmetry of shapes. Compared to previous works, our method achieves real-time performance, much more efficient than state-of-the-art (over 100 times faster). Our method also achieves higher accuracy, and is more robust.} 
\item The intrinsic symmetry problem is formulated using a functional map. To compute the entries of the functional map matrix, we design a novel deep neural network to determine the sign of each eigenfunction. Our network that predicts the sign of individual eigenfunctions using intrinsic features is compact and generalizes well to new shapes. 
\end{itemize}

%\todo{Fig.~\ref{fig:teaser} shows some examples of detecting intrinsic symmetry for shapes with complex shape/topology. These are also significantly different from examples in the training set, demonstrating the good generalization capability of our learning-based intrinsic reflectional symmetry detection technique.}

%\YL{Fig.~\ref{fig:teaser} shows some examples of detecting intrinsic reflectional symmetry for meshes with different body shapes and poses from SCAPE~\cite{anguelov2005scape}, Handstand and Swing~\cite{vlasic2008articulated} datasets. Our method is trained on an independent shape set,  demonstrating the good generalization capability of our learning-based intrinsic reflectional symmetry detection technique.}

%\qyl{As far as we have known, we are the first to develop a learning based method for detecting intrinsic reflectional symmetry, achieving real time and more robust performance.}

\section{Related Work}
\paragraph{Intrinsic Symmetry Detection.}
Many previous works cast their attention in intrinsic symmetry detection tasks. %more cite
Ovsjanikov et al.~\cite{ovsjanikov2008global} formulate the concept of intrinsic symmetry. They propose to use the Global Point Signature (GPS)~\cite{rustamov2007laplace} to transform the intrinsic symmetry of shapes into the Euclidean symmetry in the signature embedding space. The symmetry is detected by first deciding the sign sequence of eigenfunctions and then finding the nearest neighbors of the GPS of points. %This method requires much more time than ours, and the    sign flip and eigenfunction ordering
%\todo{*flip signs}
Xu et al.~\cite{xu2009partial,xu2012multi}  extend the concept of intrinsic symmetry and introduce partial symmetry where some parts of an object are symmetric. In this paper, we focus on global intrinsic symmetry due to its wide applicability, as most research in this area does.

To address the large solution space, some works parametrize intrinsic symmetry to some lower dimensional space.
A highly related problem is investigated by Mitra et al.~\cite{mitra2007symmetrization} who propose a method to symmetrize imperfectly symmetric objects. They find intrinsically symmetric point pairs by voting, and then parametrize possible transformations in a canonical space and optimize the transformation to align symmetric pairs.  
Kim et al.~\cite{kim2010mobius} use another parametrization of symmetry transformations. They find a set of symmetric points by detecting critical points of the Average Geodesic Distance 
(AGD) 
function, and generate candidate anti-M{\"o}bius transformations that can describe the symmetric transformation by enumerating subsets of the points. As a voting-based method, the running-time could be an issue. Also, the use of anti-M{\"o}bius transformation limits the method to handle genus-zero manifolds. 
Lipman et al.~\cite{lipman2010symmetry} detect symmetry by finding the orbit of points under symmetric transformations. A fuzzy point-wise symmetry correspondence matrix is generated randomly, based on which they further compute a Symmetry Factored Embedding (SFE) and Symmetry Factored Distance (SFD). However, the computation of the correspondence matrix is very time-consuming. 

The relationship between symmetry groups and matrices is studied in~\cite{lipman2010symmetry}. Similarly, Wang et al.~\cite{wang2017group} establish a homeomorphism between the symmetry group and the multiplication group of matrices. They introduce the functional map to parametrize the symmetry and limit the search space of matrix entries to the subspace of eigenfunctions. However, due to the noise in manifolds and errors during numerical calculation, eigenvalues which are ideally identical are usually calculated as different values in practice, making it difficult to determine true subspaces and resulting in poor symmetry detection. As described in \cite{wang2017group} the continuity and sparsity make functional maps a suitable representation for correspondence problems, including intrinsic symmetry. Functional maps are also used in the work~\cite{nagar2018fast}. As also mentioned in~\cite{ovsjanikov2008global}, eigenfunctions are invariant under self-isometry, apart from sign ambiguity,
and the diagonal entries of the functional map matrix are related to the sign of corresponding eigenfunctions. To decide the signs, landmark symmetric point pairs and the geodesic lines connecting them are selected. Nagar and Raman~\cite{nagar2018fast} design an explicit solution to this problem, but since their method depends on the landmark pairs, the random sampling requires a trade-off between robustness and computation complexity. 
%Compared to their method, our method is based on learning which would run faster and circumvent randomness of sampling. 
Compared to state-of-the-art methods, our learning-based method avoids explicit sampling and  is much faster (over 100 times faster for a typical example), achieving real-time performance. It circumvents the randomness of sampling, and is thus more robust and accurate.

\paragraph{Shape Analysis with Deep Learning.}
Our method learns the properties of eigenfunctions defined on manifolds using deep neural networks. 
We review 
%some previous works also study how to define network on 3D shapes. 
research that defines neural networks on 3D shapes.
With increasing requirements of faster and better analysis of 3D geometry, 
%many works nowadays 
%are dedicated to 
recent works exploit learning on shapes with deep learning.
Boscaini et al.~\cite{boscaini2016learning} design an anisotropic convolutional neural network to learn correspondences across shapes. Masci et al.~\cite{masci2015geodesic} also design a network in the spatial domain. 

Alternatively, 
another category of work constructs neural networks in the spectral domain. Bruna et al.~\cite{2013spectral} introduce a spectral convolutional layer on graphs, which can be viewed as a general form of meshes. As described in \cite{bronstein2017geometric}, a fundamental problem of spectral convolution is its dependency on the basis, making it difficult to be generalized  to different domains. To mitigate this, Yi et al.~\cite{yi2017syncspeccnn} propose a network architecture to synchronize the spectral domains and then perform convolutional operations on it. 
Rodol{\`a} et al.~\cite{Rodol2017deep} design a fully connected network to learn features that can generate functional map matrices; however, fully connected networks may suffer from overfitting, and their method requires point-wise correspondences to train the model which is not required by our method.

In this paper, we aim to detect intrinsic symmetries for general shapes, where the topology and triangulation may vary significantly. We therefore prefer a network architecture that can run robustly in cross-domain settings. To circumvent irregular connectivity of meshes, we take as input intrinsic geometric features defined on mesh vertices, namely Laplacian eigenfunctions, which implicitly carry connectivity information, but avoid coping with complex mesh connectivity. Our method thus handles general mesh topology and has good generalizability.

%directly compute on the vertices so the input to our network is essentially a point set. 
%Qi et al.~\cite{qi2017pointnet} propose PointNet to learn on point cloud data, and PointNet++~\cite{qi2017pointnet++} further adds pooling layers to improve its learning capacity. Unlike these works, our network does not take extrinsic features like point (vertex) positions as input. This is because our network is designed to detect the sign of eigenfunctions, which is intrinsic, and invariant under extrinsic transformations.
%Therefore, we instead derive input features from  eigenfunctions, such that they implicitly carry connectivity information, regardless that it is defined on an unordered point set.

\begin{figure*}[t]
	\begin{center}
		\includegraphics[width=\linewidth]{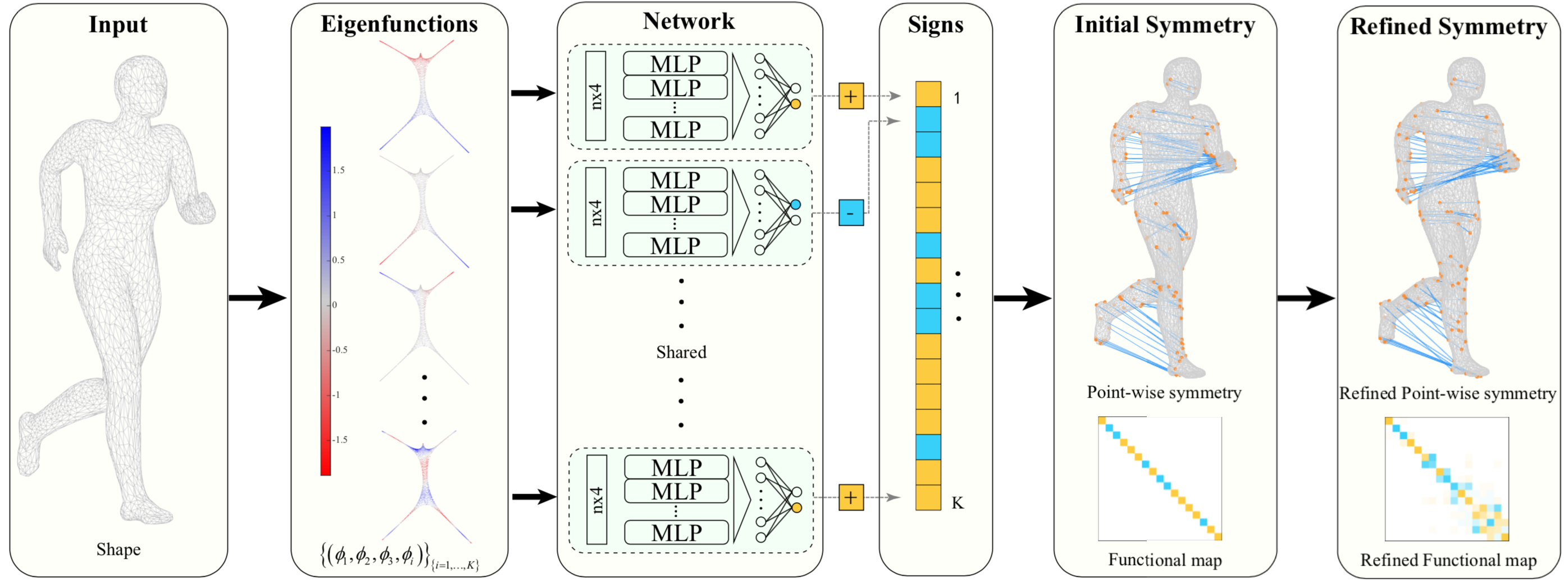}
	\end{center}
	\caption{Pipeline of our method. Given a triangle mesh of a shape, our method can predict point-wise intrinsically symmetric correspondences. First, we compute the Laplace-Beltrami eigenvectors of the mesh. Then we retain the eigenvectors associated with the first $K$ smallest eigenvalues (excluding eigenvalue 0) for analysis. We train a neural network called SignNet, that predicts the sign of each eigenfunction $\bm{\phi}_i$ ($i=1, 2, \dots, K$), under reflectional symmetry transformation $T$, i.e.,  whether $\bm{\phi}_i\circ T=\bm{\phi}_i$ or $\bm{\phi}_i\circ T=-\bm{\phi}_i$. 
 	Instead of feeding the extrinsic positions as input, we use as input the first $t$ eigenvectors along with the $i$-th eigenfunction $\bm{\phi}_i$, which are invariant under isometric transformation. The second box visualizes the shapes using the first three eigenvectors $(\bm{\phi}_1,\bm{\phi}_2,\bm{\phi}_3)$ as the coordinates and $\bm{\phi}_i$ as the color (where blue is small and red is large). The output is a two-dimensional vector, indicating the sign. We later combine $K$ signs and convert them to a diagonal functional map matrix whose entries $c_{ii}$ are either $+1$ or $-1$. This induces an initial point-wise symmetry map. To correct errors led by repeating eigenvalues and imperfect symmetry, we refine the diagonal matrix and obtain a better symmetry mapping as shown in the last box through post-processing.  }
	\label{fig:netstruct}
\end{figure*}

\section{Representing Intrinsic Symmetry by Functional Maps}\label{sec:fm}

To cope with discrete and high-dimensional point-wise correspondence matrices, we use  functional maps to represent the self-mapping.
The functional map was introduced in \cite{ovsjanikov2012functional}, first used to describe the correspondences between two shapes. In our problem, a self-isometry $T$ can also be viewed as a mapping between two identical shapes $\mathcal{M}$. And the mapping $T:\mathcal{M}\rightarrow\mathcal{M}$ can naturally introduce a bijective transformation $T_f\in GL(L^2(\mathcal{M}))$ in the square-integrable space $L^2(\mathcal{M})$, such that
\begin{equation}
	\forall f\in L^2(\mathcal{M}),m\in\mathcal{M},T_f(f)(m)=f(T(m)).
\end{equation}

Three remarks are presented in \cite{ovsjanikov2012functional} w.r.t. \textit{functional map} $T_f$, which are summarized below. %For further study of the functional map, we would also write the proof for the third remark.
\begin{prop}
The original self-isometry $T$ can be recovered from $T_f$.
\end{prop}

\begin{prop}
For each intrinsic symmetry $T$ on $\mathcal{M}$, $T_f$ is a linear transformation on $L^2(\mathcal{M})$.
\end{prop}

\begin{prop}
Assume that $L^2(\mathcal{M})$ is equipped with an orthogonal basis ${\{\bm{\phi}_i\}}_{i=1,2,\dots}$. For each $T$, the functional map can be represented by a matrix $\bm{C}$, with entries $c_{ij}=<T_f(\bm{\phi}_i),\bm{\phi}_j>$. For each function $f=\sum_i b_i\bm{\phi}_i\in L^2(\mathcal{M})$ with coefficient vector $\bm{b}=(b_1,b_2,\dots)$, the coefficient vector of map $T_f(f)$ of $\bm{b}$ is $\bm{C}\bm{b}$.
%see how to write proof
\end{prop}

%\begin{proof}
%The function $f\in L^2(\mathbf{M})$ can be projected to the basis as $f=\sum_i a_i\bm{\phi}_i $. Apply the transformation to f, and we obtain
%\begin{equation}
%	T_f(f)=T_f(\sum_i a_i\bm{\phi}_i)=\sum_i a_iT_f(\bm{\phi}_i).
%\end{equation}
%The transformation on the basis function can be written as $T_f(\bm{\phi}_i)=\sum_j C_{ij}\bm{\phi}_j$. When the basis functions are orthogonal, we can deduce that $C_{ij}=<T_f(\bm{\phi}_i),\bm{\phi}_j>$. Take this back to the previous equation, then it writes as
%\begin{equation}
%	T_f(f)=\sum_i a_iT_f(\bm{\phi}_i)=\sum_i a_i\sum_j c_{ji} \bm{\phi}_j.
%\end{equation}
%Now that the coefficient vector of $T_f(f)$ is $\mathbf{C}^T\mathbf{a}$.
%\end{proof}

Following the choice of~\cite{ovsjanikov2012functional}, we use the eigenfunctions of Laplace-Beltrami operator as the basis. 
For a mesh with $N$ vertices, the discrete Laplacian operator on the mesh is defined as an $N\times N$ matrix~\cite{Meyer2003}
\begin{equation}
	\bm{L}=\bm{A}^{-1}(\bm{D}-\bm{W}),
\end{equation}
where $\bm{A}=diag(a_1,\dots,a_N)$ contains vertex weights, with $a_i$ equal to the Voronoi area of the vertex (i.e., a third of the sum of one-ring neighborhood areas). %\todo{*expression}
$\bm{W}=\{w_{ij}\}_{i,j=1,\dots,N}$ is the sparse cotangent weight matrix, $\bm{D}$ is the degree matrix which is a diagonal matrix with diagonal entries $d_{ii}=\sum_{j=1}^{N}{w_{ij}}$. 

The aforementioned eigenfunction basis $\bm{\bm{\phi}}=\{\bm{\phi}_i\}_{i=1,2,\dots,N}$ are the solution of
$\bm{L}\bm{\bm{\phi}}=\bm{\Lambda}\bm{\bm{\phi}}$, where $\bm{\Lambda}$ is a diagonal matrix whose diagonal entries are eigenvalues in ascending order, $\lambda_0\leq\lambda_1\leq \dots \lambda_N$.
%\todo{*}
For efficiency and robustness, we take the eigenfunctions corresponding to the first $K$ smallest eigenvalues ($K << N$). Note $\lambda_0 = 0$ and the corresponding trivial eigenfunction is ignored.
%truncate, 1-ring neighbor, head note of c

\section{Method}
\subsection{Overview}
Our goal is to detect the intrinsic symmetry of shapes. An intrinsic symmetry is the self-homeomorphism of a smooth surface $\mathcal{M}$, written as $T:\mathcal{M}\rightarrow\mathcal{M}$, which preserves geodesic distances $d_g$ 
\begin{equation}\label{eq:intrsdef}
\forall m,n,T(m),T(n)\in \mathcal{M},d_g(m,n)=d_g(T(m),T(n)).
\end{equation}
Instead of directly computing a point-wise correspondence matrix, we use a functional map to describe this self-mapping.
The functional map defined on the Laplacian basis is represented as a matrix, which is the coordinate transformation matrix w.r.t. the source and target bases.
%\todo{*}
Since the Laplace-Beltrami operator is invariant under isometric transformation, the eigenfunction space stays invariant under self-mapping. Therefore, the matrix $C$ corresponding to the self-mapping $T$ is a block diagonal matrix. More specifically, only one of the two cases holds for eigenfunction $\bm{\phi}_i$ associated with non-repeating eigenvalues (see also in \cite{ovsjanikov2008global}):

\begin{itemize}
	\item $\bm{\phi}_i\circ T=\bm{\phi}_i$, where $\bm{\phi}_i$ is called 
	\textit{positive}.
	\item $\bm{\phi}_i\circ T=-\bm{\phi}_i$, where $\bm{\phi}_i$ is called \textit{negative}.
\end{itemize}

Therefore, the entry $c_{ii}$ in the matrix  corresponding to each non-repeating eigenfunction $\bm{\phi}_i$ should be either +1 or -1, depending on whether $\bm{\phi}_i$ is positive or negative. 
Fig.~\ref{fig:netstruct} shows the pipeline of our method. %*write more
We train a network called SignNet to distinguish the sign of eigenfunctions under reflectional symmetry. To provide sufficient guidance, we train the network in a supervised fashion. 
%The ground-truth signs for eigenfunctions of shapes in the training set are first computed using a state-of-the-art algorithm~\cite{nagar2018fast} and then manually checked and corrected. 
Given an input shape, once the signs of Laplacian eigenfunctions are predicted using our SignNet, 
we can fill in the diagonal of the initial functional map matrix $\tilde{C}$ with +1 and -1. However, in most of the time the intrinsic symmetry is imperfect, where some areas experience non-isometric deformation. Moreover, there could also be eigenfunction spaces associated with repeating eigenvalues, in which condition the diagonal matrix cannot fully express the mapping. Therefore we use a postprocessing step to fine-tune the initial matrix $\tilde{C}$ to obtain the final matrix $C$.

\begin{figure*}[t]
	\begin{center}
		\includegraphics[width=\linewidth]{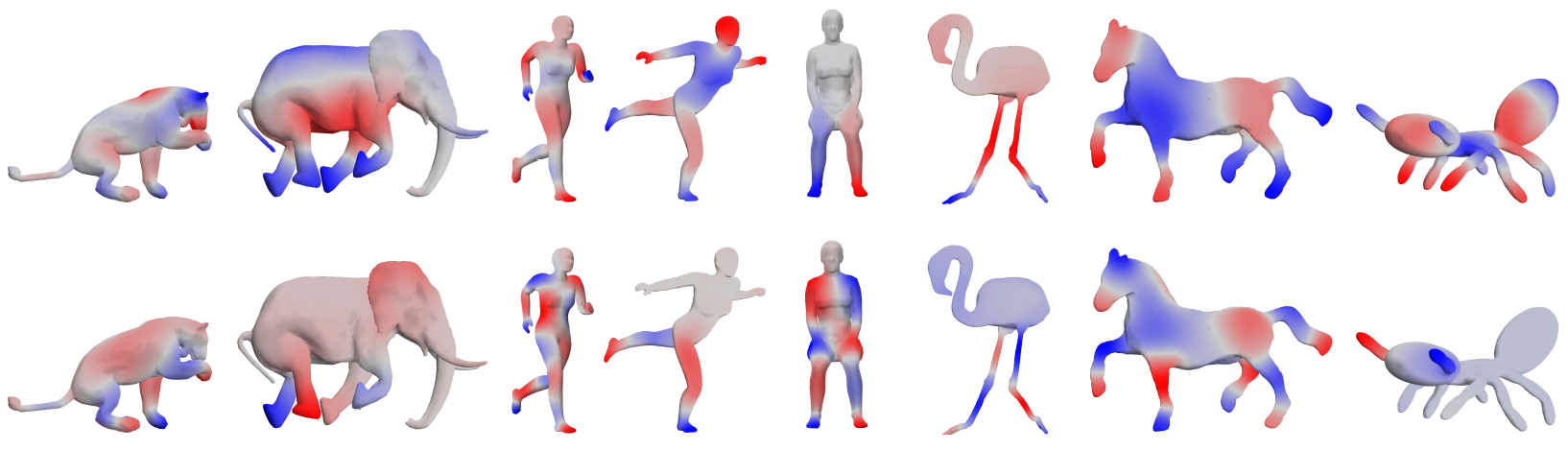}
	\end{center}
	\caption{The Laplace-Beltrami eigenfunction maps on shapes. In this figure, we visualize the eigenfunctions on shapes from SHREC 2007~\cite{giorgi2007shrec}, elephant, and flamingo~\cite{sumner2004deformation}. Red vertices represent positive values of eigenfunctions and blue values are negative values. First row shows eigenfunctions where $\bm{\phi}_i\circ T = \bm{\phi}_i$ (i.e. positive cases), and the second row presents negative cases, satisfying $\bm{\phi}_i\circ T=-\bm{\phi}_i$. Obvious symmetric/asymmetric patterns can be observed in these non-repeating eigenfunctions. We further develop a deep neural network to predict the signs.  }
	\label{fig:featuremap}
\end{figure*}

\subsection{Learning Intrinsic Symmetry}

\paragraph{Diagonal entries of the functional map matrix.}

As described in Section~\ref{sec:fm}, we detect intrinsic symmetry by computing the functional map matrix $C$. Although the dimension of $K\times K$ functional map matrix $C$ is already much lower than the $N\times N$ point-wise correspondence matrix, predicting the full $K\times K$ matrix is still challenging for optimization methods or deep networks since there are still too many variables.
%variables is still too big either for optimization methods or for deep networks to manipulate. 
We further utilize the sparse structure of the symmetry functional map to make it much easier to predict the mapping.

First we need to clarify an important property about eigenfunctions under a symmetry mapping. The Laplacian eigenfunctions associated with non-repeating eigenvalues are invariant under intrinsic symmetry mapping, only with sign ambiguity. 
%Here are the formulation and proof of this property.
This property is formally presented as follows:
\begin{theorem}\label{th:eigen}
 For an intrinsic mapping $T$ defined in Equation~(\ref{eq:intrsdef}) and a Laplacian eigenfunction $\bm{\phi}_i$ associated with a non-repeating eigenvalue $\lambda_i$, $\bm{\phi}_i\circ T=\pm \bm{\phi}_i$. 
 \end{theorem}
\begin{proof} 
 As a well-known property of Laplace-Beltrami operator, the operator $L:L^2(\mathcal{M})\rightarrow L^2(\mathcal{M})$ is invariant under isometric transformation $T$, i.e. %$T_f:L^2(\mathcal{M})\rightarrow L^2(\mathcal{M})$ 
\begin{equation}
    \forall f\in L^2(\mathcal{M}), T_f L(f)=L T_f(f),
\end{equation}
where $T_f:L^2(\mathcal{M})\rightarrow L^2(\mathcal{M})$ is the transformation on $L^2(\mathcal{M})$ introduced by $L$. 

Let $f=\bm{\phi}_i$, we can then obtain $L (T_f\bm{\phi}_i)=T_f L\bm{\phi}_i =\lambda_i T_f\bm{\phi}_i$, which means $T_f\bm{\phi}_i$ is also an eigenfunction with $\lambda_i$ as its eigenvalue. We have $L\bm{\phi}_i=\lambda_i\bm{\phi}_i$, so $T_f(\bm{\phi}_i)=\bm{\phi}_i\circ T$ is in the same eigenfunction space as $\bm{\phi}_i$. Given that $\lambda_i$ is non-repeating and $T$ is isometric, then we have $\bm{\phi}_i\circ T=\pm \bm{\phi}_i$.  %\todo{*why +-1}
 \end{proof}
From the proof of Theorem~\ref{th:eigen}, we can know that $\bm{\phi}_i\circ T$ and $\bm{\phi}_i$ are in the same eigenfunction space. In particular, if this eigenvalue is non-repeating, we can denote $\bm{\phi}_i\circ T=s_i \bm{\phi}_i$, where $s_i=<T_f(\bm{\phi}_i),\bm{\phi}_i> \in \{-1, +1\}$ . 

Based on this property of eigenfunctions, we further exploit the relationship between $\bm{\phi}_i$ and the functional map matrix $C$.

\begin{theorem} \label{th:eg}
If all eigenvalues are non-repeating, then $c_{ij}=s_i$, if $i=j$, or $0$ otherwise. 
 \end{theorem}
 
\begin{proof} 
 $c_{ij}=\ <T_f(\bm{\phi}_i),\bm{\phi}_j>$. If $i=j$, as defined in Theorem~\ref{th:eigen}, $c_{ii}=s_i$; if $i\neq j$, since $\lambda_i\neq \lambda_j$, then $c_{ij}= \ <T_f(\bm{\phi}_i), \bm{\phi}_i>\ =\ <T_f(\bm{\phi}_i), \bm{\phi}_j>\ = 0$. 
 \end{proof}
 
\qyl{Theorem~\ref{th:eg} shows that} $\bm{C}$ is a block diagonal matrix, where the  entry associated with the non-repeating eigenvalue $\lambda_i$ is $s_i$.

\paragraph{Predicting the sign of eigenfunctions.}

So the problem is much simplified and disentangled, such that we can derive the whole matrix by separately considering the sign of each eigenfunction. The visualization of the eigenfunctions on shapes is shown in Fig.~\ref{fig:featuremap}. In this illustration, red areas represent positive values and blue areas are negative values. The first row shows eigenfunctions that satisfy $\bm{\phi}_i\circ T=\bm{\phi}_i$ (i.e. positive cases), and the second row includes shapes associated with a negative eigenfunction. From the figure it can be seen that symmetric patterns are rather obvious: positive functions appear symmetric under reflectional symmetry, while negative ones are asymmetric. Nagar and Raman~\cite{nagar2018fast} propose a sampling-based method to decide the sign of the function. However, this approach depends on random samples, which takes a long time to compute and may occasionally fail. In this paper, we propose to train a neural network to learn the sign of eigenfunctions.

Fig.~\ref{fig:netstruct} illustrates the pipeline of our method. Given an input shape, we first compute its Laplacian matrix and the first $K$ eigenfunctions (excluding the trivial eigenfunction associated with eigenvalue $0$). Instead of taking the whole shape along with the eigenfunctions as input, which requires the neural network to deal with irregular mesh connectivity, our neural network (SignNet) processes each eigenfunction $\bm{\phi}_i$ \emph{separately}. Assuming the $i$-th eigenfunction $\bm{\phi}_i$ is being processed,  the input to the network includes not only $\bm{\phi}_i$, but also the first $t$ eigenfunctions $\bm{\phi}_1,\bm{\phi}_2, \dots, \bm{\phi}_t$, which capture the characteristics of the input mesh and are also intrinsic.

The output of SignNet is a 2-dimensional softmax vector. The distributions of the eigenfunctions on the mesh can reflect the pattern of the sign to a great extent. Here we do not use the original positions of vertices as input since they are extrinsic features. In contrast, the first $t$ dimensions of Laplacian eigenvectors are intrinsic, thus more suitable for detecting intrinsic symmetry. 

To visualize this, in Fig.~\ref{fig:netstruct}(b), %update fig 2 with abcd later
we plot the embedding of vertices taking the first three eigenfunctions $(\bm{\phi}_1(p),\bm{\phi}_2(p),\bm{\phi}_3(p))$ evaluated at vertex $p$ as vertex coordinates and $\bm{\phi}_i(p)$ as the color (blue to red means small value to large value). It can be observed that the shapes of the embedding are extrinsically symmetric even if the mesh is only intrinsically symmetric. Also, we can see that those eigenfunctions are either symmetric or asymmetric, corresponding to positive or negative eigenfunctions.

In the SignNet neural network, we use Multi-Layer Perceptrons (MLPs) to extract vertex features with increasing complexity. Then a max-pooling is applied on all vertices to aggregate global features. Following the pooling layers are several fully-connected layers with decreasing numbers of channels. In the end, the network predicts a two-dimensional score vector $\bm{v}_i = (v_{i,1}, v_{i,2})$, i.e.
\begin{equation}
\bm{v}_i=SignNet(\bm{\phi}_1,\bm{\phi}_2, \dots, \bm{\phi}_t; \bm{\phi}_i),
\end{equation}
such that the sign is predicted to be negative if $\argmax_k v_{i,k}=1$, or positive if $\argmax_k v_{i,k}=2$, for $k=1, 2$. Let 
$\bm{\hat{s}}_i$ be a two-dimensional vector, 
$\bm{\hat{s}_i} = (1, 0)$ if $s_i = -1$, and $\bm{\hat{s}_i} = (0, 1)$ if $s_i = 1$.
The loss function is designed as the \textit{CorssEntropy} between $\bm{v}_i$ and ground-truth sign label $\bm{\hat{s}}_i$, formulated as
\begin{equation}
Loss = CrossEntropy(\bm{v}_i,\bm{\hat{s}}_i).
\end{equation}
%The ground truth is obtained by first applying a state-of-the-art method \cite{nagar2018fast} to shapes in the training set, and then manually removing shapes that 

\YL{
\subsection{Training Data}\label{sec:training}
Our learning-based approach requires a dataset for training. 
For this purpose, we choose as training set a fusion of sets SHREC 2007~\cite{giorgi2007shrec}, elephant, and flamingo~\cite{sumner2004deformation}, which contains non-rigidly deformed shapes which are intrinsically symmetric. 
As a shape retrieval dataset, SHREC dataset includes shapes of  different categories. 
Meanwhile, they are also independent from the test sets (SCAPE~\cite{anguelov2005scape} and TOSCA~\cite{bronstein2008numerical}). This ensures fairness and evaluates the generalizability of our learning-based approach. 

We built a simple user interface to visualize and manually label each Laplacian eigenfunction as either positive, negative or neither, under reflectional symmetry transform.
Neither cases happen for shapes which are not intrinsically reflectional symmetry, or for eigenfunctions with repeating eigenvalues, as shown in Fig.~\ref{fig:nosign}. These are excluded from our training dataset.
The dataset will be released to the community to facilitate future research.
}

\begin{figure}[t]
	\begin{center}
		\includegraphics[width=\linewidth]{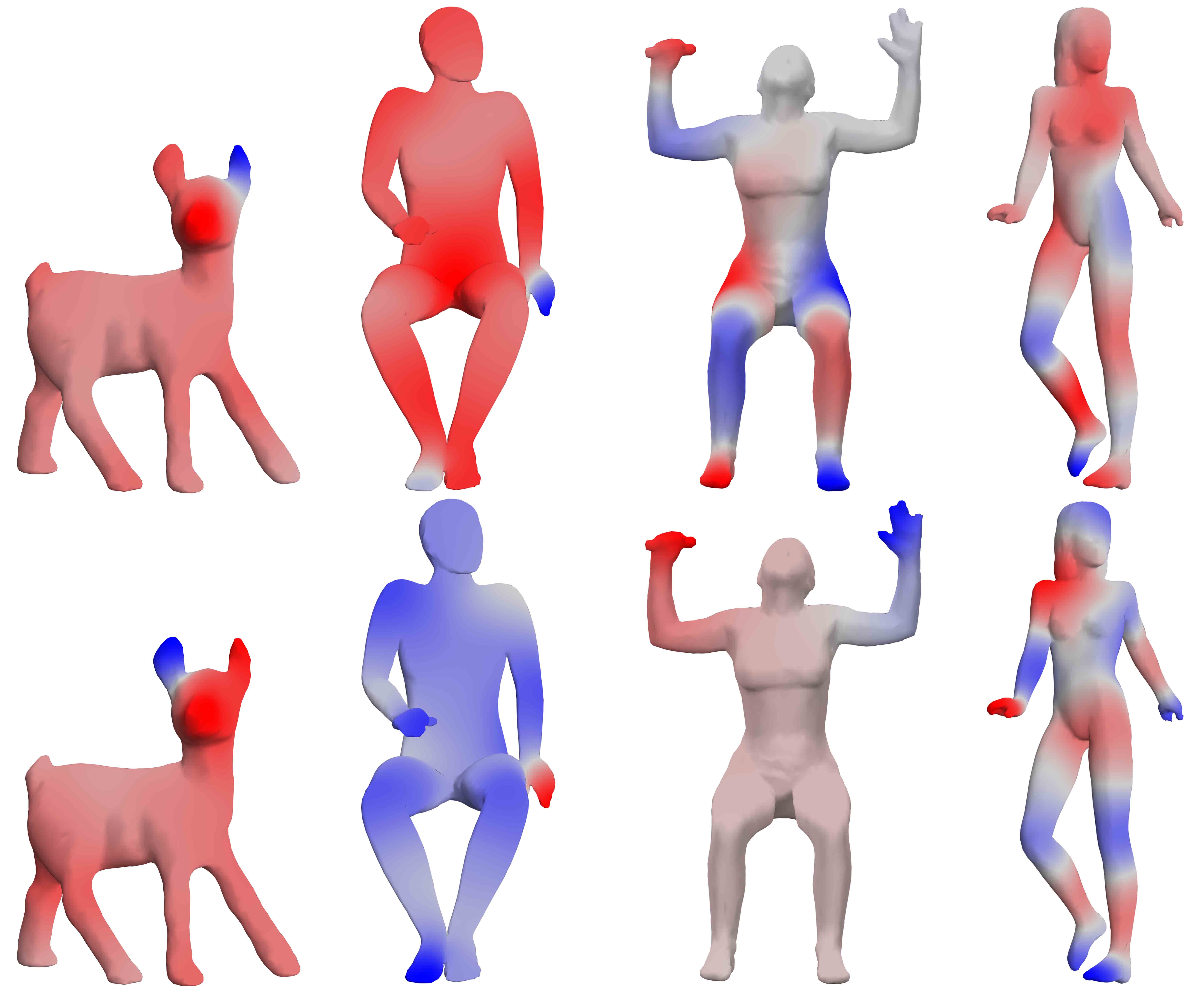}
	\end{center}
	\caption{Eigenfunctions which are neither positive nor negative. In this figure, we visualize eigenfunctions associated with repeating eigenvalues, such that $\bm{\phi}_i\circ T\neq \pm\bm{\phi}_i$, where $T$ is the intrinsic reflectional symmetry mapping.}
	\label{fig:nosign}
\end{figure}

\subsection{Network Architecture}

%%%YKL Add a figure TODO
In our SignNet, the input placeholder is set to work with 4500 points, which are padded with 0 if the mesh has less than 4500 vertices, and for meshes with more than 4500 vertices,  they are downsampled to 4500 points. In the network, we use multi-layer perceptrons (MLP), max-pooling layers, and fully connected layers. There are five MLP layers, having 64, 128, 256, 512, 4096 channels respectively, and there are ReLU activation layers and batch normalization layers right after the output of each MLP layer. Then we use a max-pooling layer to aggregate the global features. Such a combination of shared-weight MLP and max-pooling layers are proven to be effective to fit functions defined on the point set (see the appendix in~\cite{qi2017pointnet}). Then, four fully connected layers are applied to the global features. Their output channels are 512, 128, 32, 2. The first three layers are also connected with ReLU activation, batch normalization, and (70\%) dropout layers. 

\subsection{Post-processing}
In most of the time, the meshes that we are processing are not perfectly intrinsically symmetric. The entries of functional matrices would not be exactly -1 and +1. Moreover, owing to the imperfect triangulation and discretization of Laplacian operator, in numerical computation, eigenvalues are mostly non-repeating, but there are actually eigenfunction spaces with multiple eigenfunctions. Therefore, the entries associated with sub-eigenfunction spaces need more entries, usually in the form of an orthogonal sub-matrix, to describe the functional map. 

In consideration of the above two reasons, we use a post-processing step to correct the matrix and convert it to point-wise correspondence. Similar to~\cite{nagar2018fast}, we use functional constraints~\cite{ovsjanikov2012functional} to align the correspondences. As shown in Fig.~\ref{fig:self1}, this can correct initially imprecise mappings.

\section{Results and Evaluation}\label{sec:exp}
We first describe the implementation details of our method in Section~\ref{sec:imp}. 
In Section~\ref{sec:cmp} we compare our method with existing methods, both qualitatively and quantitatively. In addition to the accuracy of symmetry, we also measure the run time of different methods, showing the significant superiority of our method efficiently. We further test the robustness of our method in Section~\ref{sec:rob}. Due to the shared-weight structure of our network, our method stays robust under different topology and vertex numbers.

\subsection{Implementation Details}\label{sec:imp}

We now present details of the training and test process of our SignNet.

The computation of Laplacian matrix and eigenvectors are described in Section~\ref{sec:fm}. Please refer to~\cite{Meyer2003} for more implementation details related to these steps. 
We implement the neural network architecture with Tensorflow.
The network is optimized using Adam~\cite{kingma2014adam} solver. The initial learning rate is set to $1\times 10^{-4}$ and momentum is 0.9. We choose to truncate at first 12 lowest eigenvectors (i.e., $K=12$), and by default the input feature has 4 dimensions, composed of first 3 eigenvectors (i.e., $t=3$) and the $i$-th eigenvector. We train the network for 500 epochs on a PC with an NVIDIA 1080TI GPU and an intel i7-7700 CPU.  

\begin{table}[t]
	\begin{center}
		\setlength{\abovecaptionskip}{2pt}
		\caption{Comparison of correspondence rate, mesh rate, and running time on the SCAPE dataset. We compare our method with MT~\cite{kim2010mobius}, BIM~\cite{kim2011blended}, OFM~\cite{liu2015properly}, GRS~\cite{wang2017group}, and FA~\cite{nagar2018fast}.}	
		\begin{adjustbox}{max width=0.46\textwidth}
		\begin{tabular}{c|cccccc}
			\hline
			&MT&BIM&OFM&GRS&FA&Ours \\ \hline
			Corr. Rate(\%)&82.0&84.8&91.7&94.5&97.5&98.1 \\
			Mesh Rate(\%)&71.8&76.1&97.2&98.6&100&100 \\
			Time(s)&18.0&304.26&50.70&20.28&6.77&0.06 \\  \hline
		\end{tabular}	        
		\end{adjustbox}
		\label{tab:scape_comp}
	\end{center}
\end{table}

\begin{table}[t]
	\begin{center}
		\setlength{\abovecaptionskip}{2pt}
		\caption{Correspondence rate (\%) comparison on TOSCA.}
		\begin{tabular}{c|cccccc}
			\hline
			&MT&BIM&OFM&GRS&FA&Ours \\ \hline
			Cat&66.0&93.7&90.0&96.5&95.6&96.0 \\
			Centaur&92.0&100&96.0&92.0&100&100 \\
			David&82.0&97.4&94.8&92.5&96.2&97.2 \\
			Dog&91.0&100&93.2&97.4&98.8&100 \\
			Horse&92.0&97.1&95.2&99.5&97.3&96.4 \\
			Michael&87.0&98.9&94.6&91.4&96.5&98.7 \\
			Victoria&83.0&98.3&98.7&95.5&96.2&97.8 \\
			Wolf&100&100&100&100&100&100 \\
			Gorilla&-&98.9&98.9&100&100&100 \\  \hline
			Average&85.0&98.0&95.1&94.5&97.8&98.1 \\  \hline
		\end{tabular}		        
		\label{tab:tosca_cr}
	\end{center}
\end{table}

\subsection{Comparison of Results}\label{sec:cmp}
As one of the biggest advantages of learning-based methods, our algorithm runs much faster than previous sampling based intrinsic symmetry detection algorithms. Also, the neural network can learn some common properties of eigenfunctions across models to distinguish the sign of eigenfunctions. This would avoid randomness of sampling, so also has better performance in terms of correspondence accuracy. In this section, we compare our method with state-of-the-art methods including MT~\cite{kim2010mobius}, BIM~\cite{kim2011blended}, OFM~\cite{liu2015properly}, GRS~\cite{wang2017group}, and FA~\cite{nagar2018fast} on the following three metrics, widely used in the literature: 
\begin{enumerate}
\item \textit{Correspondence rate}: Assume that $(m,m'), m,m'\in\mathcal{M}$ is a ground truth correspondence pair, and the algorithm's prediction is $(m,T(m))$. If the geodesic distance $d_g(m',T(m))$ between $m'$ and $T(m)$ is less than the threshold $\sqrt{\frac{area(\mathcal{M})}{20}}$, then we count this point as a correct matching. Correspondence rate measures the ratio of labeled points that are correctly matched. 
\item \textit{Mesh rate}: It measures the percentage of meshes whose correspondence rate is above the threshold $\beta$. We use $\beta=75\%$, the same as~\cite{nagar2018fast,wang2017group}.
\item \textit{Time}: We measure the average run time of each algorithm to compute the symmetry.
\end{enumerate}

We compare different methods using SCAPE~\cite{anguelov2005scape} and TOSCA~\cite{bronstein2008numerical} datasets which contain intrinsically symmetric meshes, \qyl{and the ground truth symmetric correspondences are from~\cite{anguelov2005correlated}.} We also test our method on Handstand, Swing~\cite{vlasic2008articulated} and FAUST~\cite{bogo2014faust} datasets for qualitative evaluation, as no ground truth correspondences are available.
\YL{As we mentioned in Section~\ref{sec:training}, our training set is independent from the test sets, to ensure fairness.}
%Since our method requires training, to ensure fairness and demonstrate the generalizability, we take shapes from \qyl{a fusion of independent sets SHREC 2007~\cite{giorgi2007shrec}, elephant, and flamingo~\cite{sumner2004deformation}} for training our network, and apply this to test on these new shape sets. 

The results on the SCAPE dataset of deformed human shapes are reported in Table~\ref{tab:scape_comp}. As can be seen, our method achieves the best accuracy: improving the correspondence rate from the previous best $97.5\%$ (FA) to $98.1\%$.
%, a significant drop of incorrect matchings by $84\%$ 
Both our method and FA achieve $100\%$ mesh correct rate. In terms of runtime, our method is over 100 times faster than FA, and even more than other existing methods. 

The results on the TOSCA dataset are reported in Tables~\ref{tab:tosca_cr} and~\ref{tab:tosca_mr} for the comparisons of correspondence rate and mesh rate, respectively. We report performance on individual object categories, and the overall average. Our method has similar improvements compared with existing methods. Some qualitative comparison is shown in Fig.~\ref{fig:comp}.

\begin{figure}[t]
	\begin{center}
		\includegraphics[width=\linewidth]{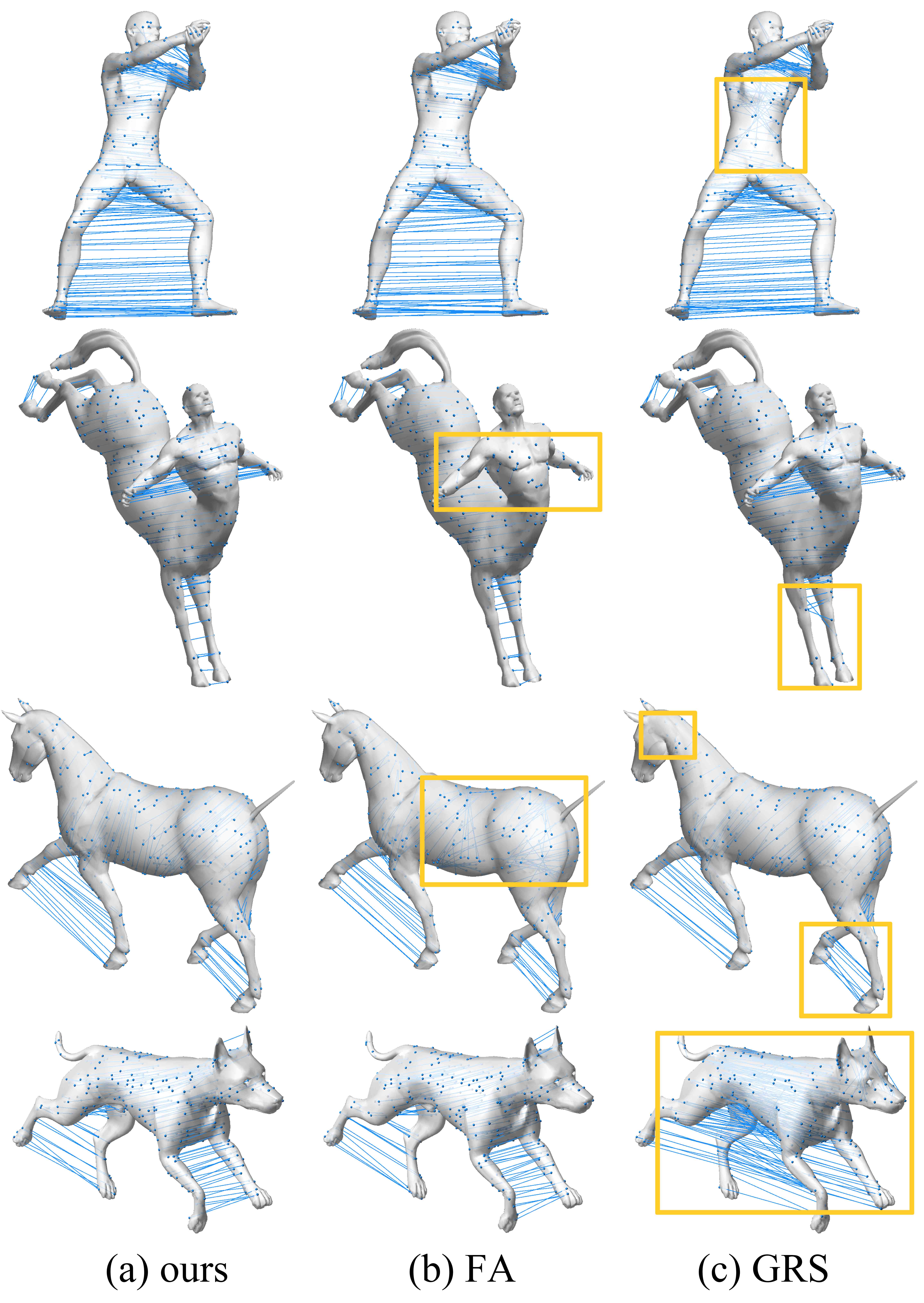}
	\end{center}
	\caption{\qyl{Qualitative comparison with previous work. (a) is symmetry predicted by our method; (b) is from FA~\cite{nagar2018fast}; (c) is result of GRS~\cite{wang2017group}. We can see that our method has the least artifacts when detecting symmetry.}}
	\label{fig:comp}
\end{figure}

\begin{table}[t]
	\begin{center}
		\setlength{\abovecaptionskip}{2pt}
		\caption{Mesh rate (\%) comparison on TOSCA.}	
		\begin{tabular}{c|cccccc}
			\hline
			&MT&BIM&OFM&GRS&FA&Ours \\ \hline
			Cat&54.6&90.9&90.9&100&100&100 \\
			Centaur&100&100&100&100&100&100 \\
			David&57.1&100&100&100&100&100 \\
			Dog&88.9&100&88.9&100&100&100 \\
			Horse&100&100&87.5&100&100&100 \\
			Michael&75&100&100&100&100&100 \\
			Victoria&63.6&100&100&100&100&100 \\
			Wolf&100&100&100&100&100&100 \\
			Gorilla&-&100&100&100&100&100 \\  \hline
			Average&76&98.7&92.6&100&100&100 \\  \hline
		\end{tabular}	        
		\label{tab:tosca_mr}
	\end{center}
\end{table}

\subsection{Evaluation of Design Choices}\label{sec:eva}

As we said before, by default we use the first three Laplacian eigenfunctions as the coordinates to embed vertices into an intrinsic space. Compared to Laplacian eigenfunctions, the raw positions are not invariant under global rigid transformation, nor under non-rigid  isometric deformation, so not suitable for predicting the sign of eigenfunctions on the mesh. In this experiment, we compare using the positions $(\bm{x},\bm{y},\bm{z},\bm{\phi}_i)$ versus eigenfunctions $(\bm{\phi}_1,\bm{\phi}_2,\bm{\phi}_3,\bm{\phi}_i)$ as input. Table~\ref{tab:design} lists the average accuracy of sign prediction on TOSCA and SCAPE datasets. It shows that the accuracy using the position (denoted as \textit{Pos.}) is much lower than that of our design. During experiments, we observe that when models have scales in a large range, the network with position input performs even worse.

We compute the functional map matrix by independently predicting the sign of eigenfunctions. This strategy circumvents the flip of signs and permutation of eigenfunctions. To show the advantage of this strategy, we design another network which takes all the eigenfunctions as input and predicts the whole $K$ diagonal entries at once. We denote this alternative design as \textit{Diag.} in Table~\ref{tab:design}. We can see the accuracy of sign prediction is much lower than ours. This is probably  because the input and output dimensions are too high for the network to learn.

The input to the network is the first $t$ eigenfunctions as well as the $i$-th eigenfunction, i.e., $[\bm{\phi}_1,...,\bm{\phi}_t; \bm{\phi}_i]$. Too small the $t$ value would make different vertices indistinguishable, impossible to determine the sign. And if $t$ is too big, it would make the network more complex, and introduce more redundant noisy high-frequency eigenvectors. Here we vary $t$ from 2 to 4. The table shows that $t=3$ (Ours) achieves the best performance.
Our input is defined on vertices.
Although existing point-based deep learning methods such as PointNet~\cite{qi2017pointnet} take extrinsic point coordinates as input,  it is possible to feed the same input to such architectures for prediction. 
We also test this by feeding our input directly to PointNet~\cite{qi2017pointnet}, and report the accuracy of sign prediction. The performance is also lower than that of our method. This is probably due to our compact network design that generalizes well to new data. 

\begin{table}[h]
	\begin{center}
		\setlength{\abovecaptionskip}{2pt}
		\caption{Evaluation of our design choices. We compare the accuracy (=number of correctly predicted signs/total number of eigenfunctions) of different design choices.}	
		\begin{adjustbox}{max width=0.46\textwidth}
		\begin{tabular}{c|cccccc}
			\hline
			&Ours&Pos.&Diag.&2 Eig.&4 Eig.&PointNet\\ \hline
			
			Acc.(\%)&98.4&60.6&66.1&96.3&94.1&96.2 \\  \hline
		
		\end{tabular}	        
		\end{adjustbox}
		\label{tab:design}
	\end{center}
\end{table}

\subsection{Robustness}\label{sec:rob}
%In this part, we show that our method is robust under different kinds of disturbance. 

We now test the robustness of our learning-based approach.

\paragraph{Different topology.} 
Since the geodesic distance and the eigenfunctions are defined on the manifold $\mathcal{M}$, the topology of $\mathcal{M}$ would contribute significantly to the computation of intrinsic symmetry. For example, MT~\cite{kim2010mobius} requires the topology to be genus-zero. In our method, since the eigenfunctions can  work consistently under different topology, the network can stay robust with topological changes. As shown in Fig.~\ref{fig:self1}, we reconstruct those meshes with self-intersection in space and %some adjacent parts are stuck together. 
the produced meshes are high-genus. 
The first row shows the original shapes with problematic regions highlighted. The second row shows the initial  correspondences of intrinsic symmetry mapping, and the correspondences after refinement. For those challenging cases, intrinsic symmetry is no longer precisely satisfied, and the refinement is effective in improving detected symmetry.

\begin{figure}[t]
	\begin{center}
		\includegraphics[width=\linewidth]{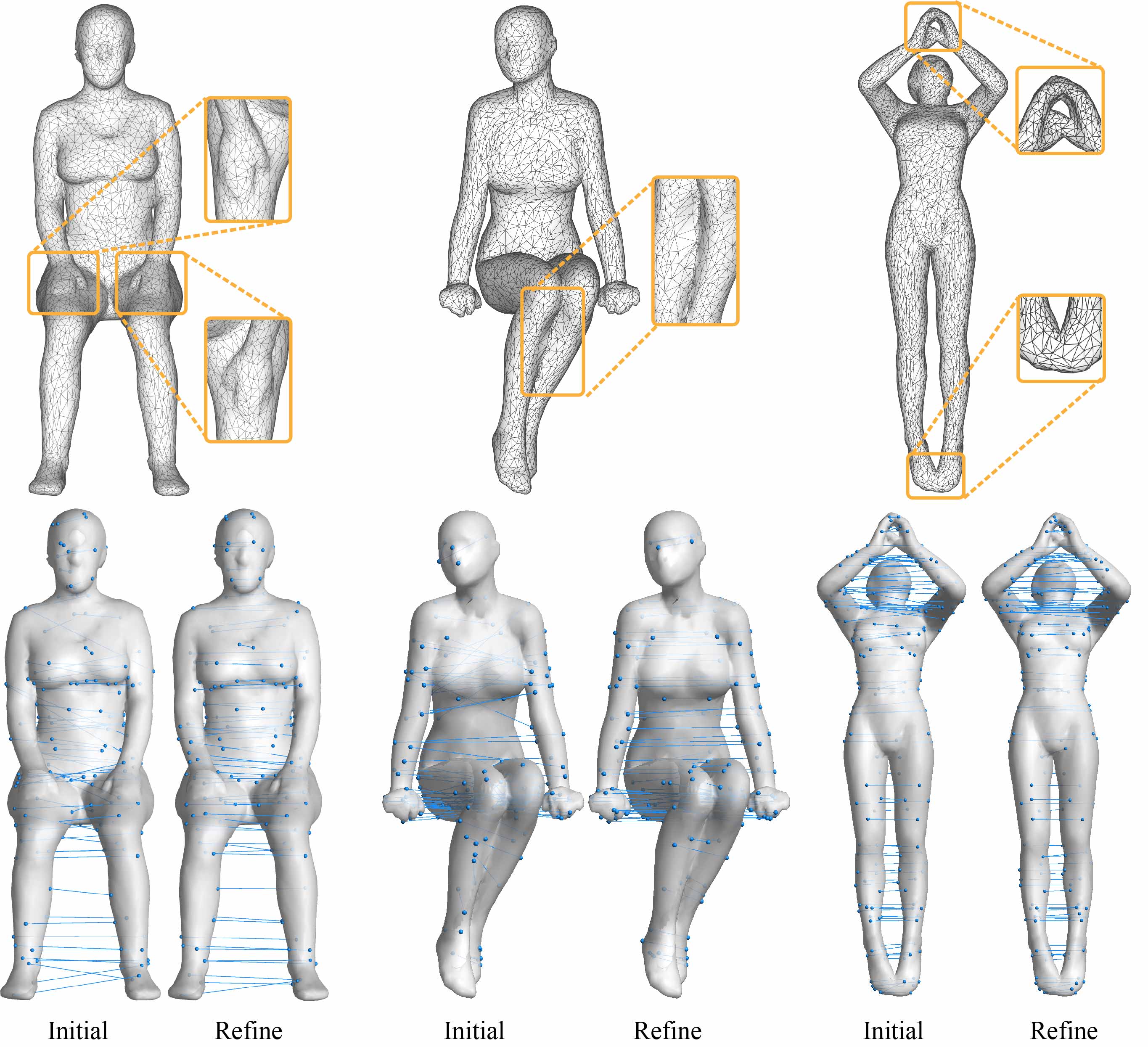}
	\end{center}
	\caption{Test with different topology. In this figure, we change the topology of models from SHREC and TOSCA to test the robustness of our method in difficult topology. We reconstruct the meshes by sticking spatially adjacent faces on the meshes together, so the shapes are no longer genus-zero. Meanwhile, the intrinsic symmetry correspondences obtained by our method are still reasonable. We also show both the initial symmetry and the refined symmetry in the second row. Since the reconstructed meshes are changed and are no longer perfectly intrinsically symmetric, the refinement step is important to polish the correspondences.}
	\label{fig:self1}
\end{figure}
% figure detail

%\textbf{Different number of vertices.}
%In this part, we show that our method can deal with different number of vertices. Compared to the fully connected network used in \cite{litany2017deep}, our network is composed of shared-weight multi-layer perceptron and max-pooling layers, which would not be disturbed by the number of points and the order of input points. 

%figure: different number. Figure, mesh corr

\paragraph{Incomplete shapes.}
Sometimes there could be \qyl{missing data} on shapes due to imperfect scanning or mesh modeling. We expect an intrinsic symmetry detection algorithm to work on such incomplete shapes. We perform a test by making some holes on the surface of the models. Fig.~\ref{fig:hole} shows the results of our method. It can be seen that the symmetry pairs on the shapes are still reasonable. 

\begin{figure}[t]
	\begin{center}
		\includegraphics[width=\linewidth]{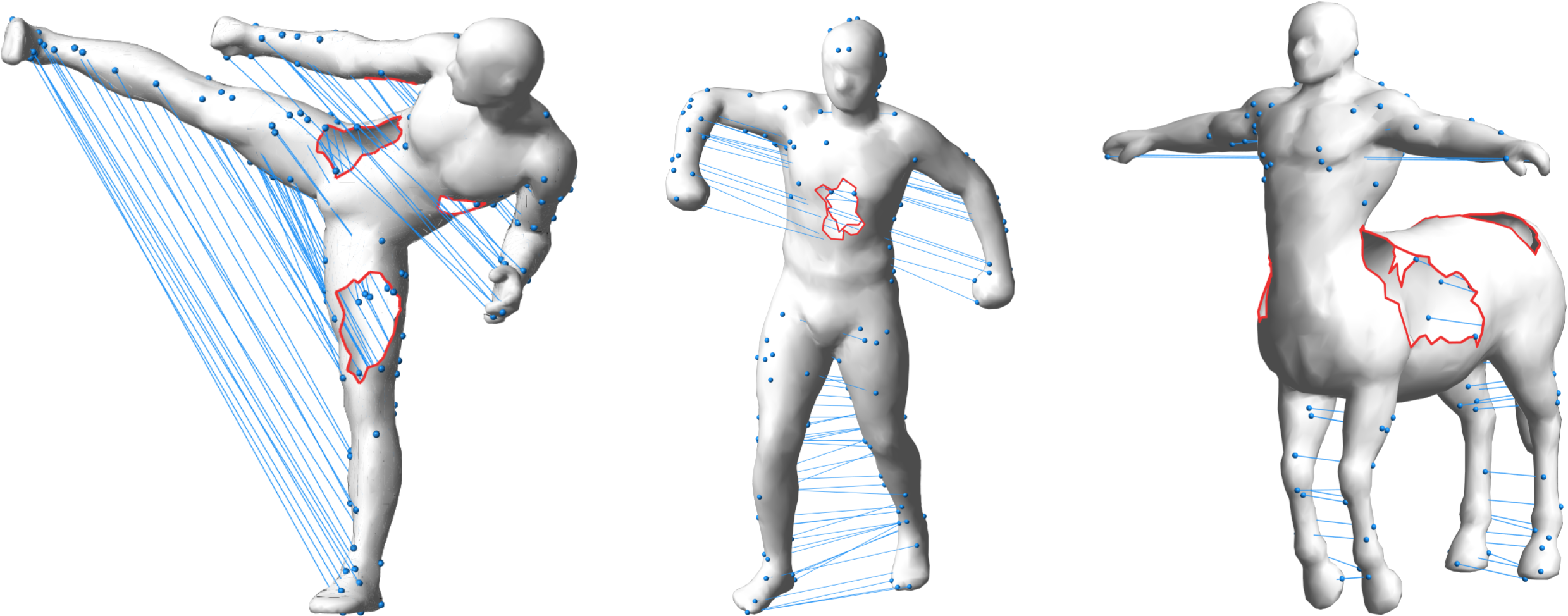}
	\end{center}
	\caption{Symmetry detection on incomplete surfaces.}
	\label{fig:hole}
\end{figure}
%figure

\subsection{Failure Case}
\qyl{As shown by the statistics, our method works well in most of cases. 
%, we show some failure cases of our method. 
However, due to the deterministic network structure, our method can only predict one symmetry result for a certain object, even if it has multiple intrinsic symmetries. In Fig.~\ref{fig:failure}, we can see that, the table has more than one reflectional symmetry plane, while our method cannot predict all of them. It would be our future work to extend our method to predict the entire symmetry group end-to-end.}

\begin{figure}[t]
	\begin{center}
		\includegraphics[width=\linewidth]{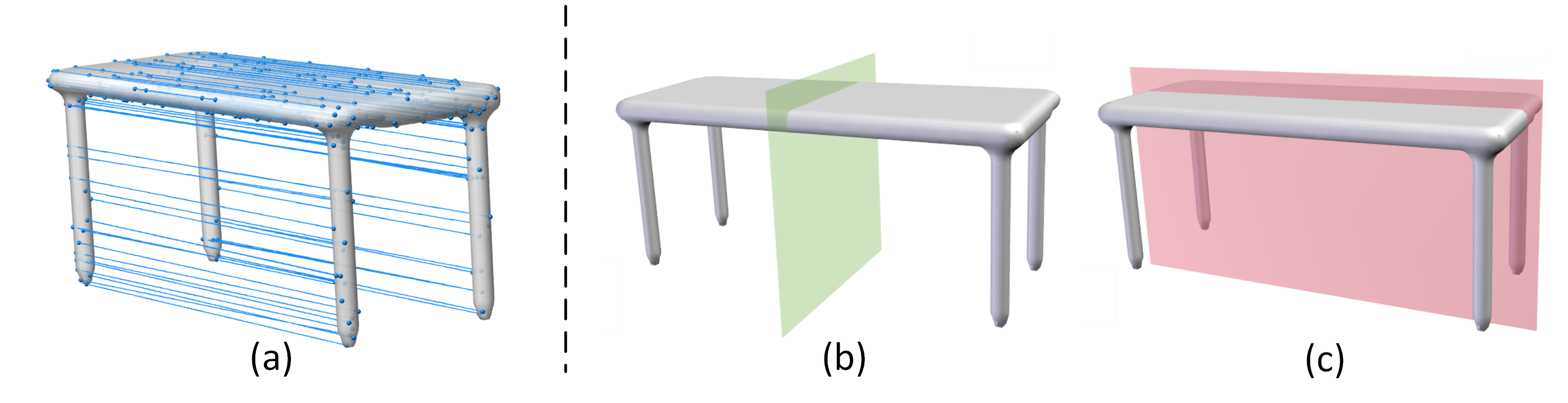}
	\end{center}
	\caption{\YL{Failure case of our method. (a) is the symmetry detection result of the $151$st model in SHREC 2007~\cite{giorgi2007shrec} computed by our method, (b) is the symmetry plane corresponding to our result, and (c) is another possible symmetry plane.}}
	\label{fig:failure}
\end{figure}
%figure

 \subsection{More Qualitative Results}\label{sec:more}
Fig.~\ref{fig:more} shows more results, \qyl{from TOSCA and FAUST~\cite{bogo2014faust} datasets.} We can see that our method generalizes well to various kinds of shapes and poses.
%\qyl{last row are models from SHREC07, but not in the training set.}
%results on more dataset

\begin{figure*}[t]
	\begin{center}
		\includegraphics[width=\linewidth]{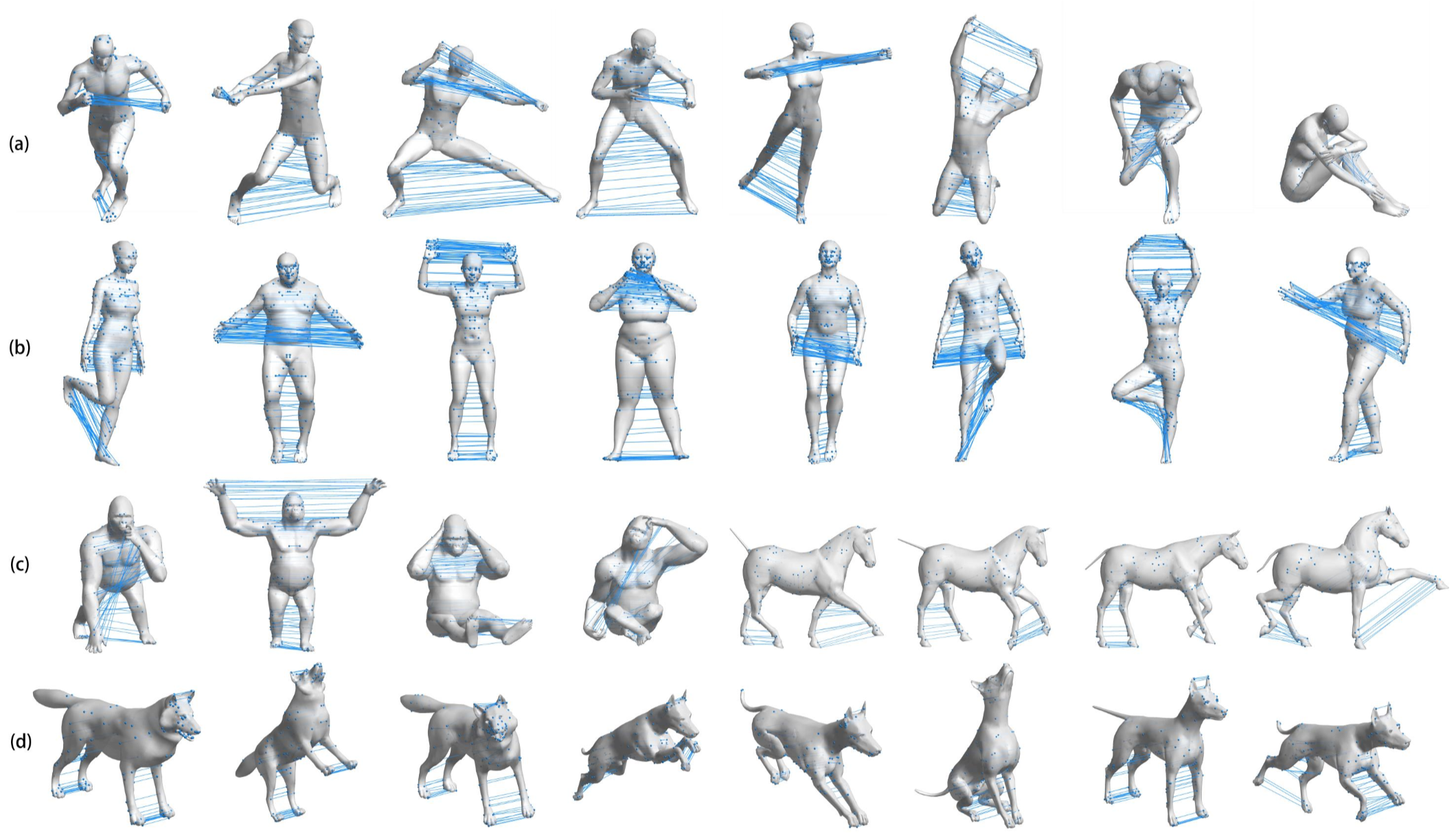}
	\end{center}
	\caption{More qualitative results. In this part, we show that our method generalizes well to various shapes including male, female, gorilla, horse, wolf, and dog~\cite{bronstein2008numerical,bogo2014faust}.}
	\label{fig:more}
\end{figure*}

\section{Conclusions and Future Work}\label{sec:conc}

In this paper, we presented a novel learning-based approach to intrinsic reflectional symmetry detection. Our method is based on functional maps and further develops a neural network architecture that predicts the sign of a Laplacian eigenfunction at a time. We design the network to take first few Laplacian eigenfunctions, in addition to the eigenfunction to be predicted. Extensive experiments show the real-time performance and superior accuracy compared with state-of-the-art methods. We also performed experiments to validate design choices and robustness of our method in challenging cases.

This work addresses global intrinsic reflectional symmetry, which is most common in practice. As future work, it would be interesting to also include rotational symmetry detection, although the property of rotational symmetry functional map matrix is more complicated.
%, it would be promising to study a uniform framework to describe both symmetry types.
Another possible direction is to extend this learning-based algorithm to partial symmetry detection. 

{\small
\bibliographystyle{ieee_fullname}
\bibliography{egbib}

\begin{thebibliography}{10}\itemsep=-1pt

\bibitem{anguelov2005scape}
Dragomir Anguelov, Praveen Srinivasan, Daphne Koller, Sebastian Thrun, Jim
  Rodgers, and James Davis.
\newblock Scape: shape completion and animation of people.
\newblock In {\em ACM transactions on graphics (TOG)}, volume~24, pages
  408--416. ACM, 2005.

\bibitem{anguelov2005correlated}
Dragomir Anguelov, Praveen Srinivasan, Hoi-Cheung Pang, Daphne Koller,
  Sebastian Thrun, and James Davis.
\newblock The correlated correspondence algorithm for unsupervised registration
  of nonrigid surfaces.
\newblock In {\em Advances in neural information processing systems}, pages
  33--40, 2005.

\bibitem{bogo2014faust}
Federica Bogo, Javier Romero, Matthew Loper, and Michael~J Black.
\newblock Faust: Dataset and evaluation for 3d mesh registration.
\newblock In {\em Proceedings of the IEEE Conference on Computer Vision and
  Pattern Recognition}, pages 3794--3801, 2014.

\bibitem{boscaini2016learning}
Davide Boscaini, Jonathan Masci, Emanuele Rodol{\`a}, and Michael Bronstein.
\newblock Learning shape correspondence with anisotropic convolutional neural
  networks.
\newblock In {\em Advances in Neural Information Processing Systems}, pages
  3189--3197, 2016.

\bibitem{bronstein2008numerical}
Alexander~M Bronstein, Michael~M Bronstein, and Ron Kimmel.
\newblock {\em Numerical geometry of non-rigid shapes}.
\newblock Springer Science \& Business Media, 2008.

\bibitem{bronstein2017geometric}
Michael~M Bronstein, Joan Bruna, Yann LeCun, Arthur Szlam, and Pierre
  Vandergheynst.
\newblock Geometric deep learning: going beyond {E}uclidean data.
\newblock {\em IEEE Signal Processing Magazine}, 34(4):18--42, 2017.

\bibitem{2013spectral}
Joan Bruna, Wojciech Zaremba, Arthur Szlam, and Yann LeCun.
\newblock Spectral networks and locally connected networks on graphs.
\newblock {\em arXiv preprint arXiv:1312.6203}, 2013.

\bibitem{dessein2017symmetry}
Arnaud Dessein, William~AP Smith, Richard~C Wilson, and Edwin~R Hancock.
\newblock Symmetry-aware mesh segmentation into uniform overlapping patches.
\newblock In {\em Computer Graphics Forum}, volume~36, pages 95--107. Wiley
  Online Library, 2017.

\bibitem{giorgi2007shrec}
Daniela Giorgi, Silvia Biasotti, and Laura Paraboschi.
\newblock Shrec: shape retrieval contest: Watertight models track.
\newblock {\em Online]: http://watertight. ge. imati. cnr. it}, 7, 2007.

\bibitem{kim2010mobius}
Vladimir~G Kim, Yaron Lipman, Xiaobai Chen, and Thomas Funkhouser.
\newblock M{\"o}bius transformations for global intrinsic symmetry analysis.
\newblock In {\em Computer Graphics Forum}, volume~29, pages 1689--1700. Wiley
  Online Library, 2010.

\bibitem{kim2011blended}
Vladimir~G Kim, Yaron Lipman, and Thomas Funkhouser.
\newblock Blended intrinsic maps.
\newblock In {\em ACM Transactions on Graphics (TOG)}, volume~30, page~79. ACM,
  2011.

\bibitem{kingma2014adam}
Diederik~P Kingma and Jimmy Ba.
\newblock Adam: A method for stochastic optimization.
\newblock {\em arXiv preprint arXiv:1412.6980}, 2014.

\bibitem{lipman2010symmetry}
Yaron Lipman, Xiaobai Chen, Ingrid Daubechies, and Thomas Funkhouser.
\newblock Symmetry factored embedding and distance.
\newblock In {\em ACM Transactions on Graphics (TOG)}, volume~29, page 103.
  ACM, 2010.

\bibitem{liu2015properly}
Xiuping Liu, Shuhua Li, Risheng Liu, Jun Wang, Hui Wang, and Junjie Cao.
\newblock Properly constrained orthonormal functional maps for intrinsic
  symmetries.
\newblock {\em Computers \& Graphics}, 46:198--208, 2015.

\bibitem{masci2015geodesic}
Jonathan Masci, Davide Boscaini, Michael Bronstein, and Pierre Vandergheynst.
\newblock Geodesic convolutional neural networks on {R}iemannian manifolds.
\newblock In {\em Proceedings of the IEEE international conference on computer
  vision workshops}, pages 37--45, 2015.

\bibitem{Meyer2003}
Mark Meyer, Mathieu Desbrun, Peter Schr\"{o}der, and Alan~H. Barr.
\newblock Discrete differential-geometry operators for triangulated
  2-manifolds, 2003.

\bibitem{mitra2007symmetrization}
Niloy~J Mitra, Leonidas~J Guibas, and Mark Pauly.
\newblock Symmetrization.
\newblock In {\em ACM Transactions on Graphics (TOG)}, volume~26, page~63. ACM,
  2007.

\bibitem{mitra2014structure}
Niloy~J Mitra, Michael Wand, Hao Zhang, Daniel Cohen-Or, Vladimir Kim, and
  Qi-Xing Huang.
\newblock Structure-aware shape processing.
\newblock In {\em ACM SIGGRAPH 2014 Courses}, page~13. ACM, 2014.

\bibitem{mitra2010illustrating}
Niloy~J Mitra, Yong-Liang Yang, Dong-Ming Yan, Wilmot Li, Maneesh Agrawala,
  et~al.
\newblock Illustrating how mechanical assemblies work.
\newblock 2010.

\bibitem{nagar2018fast}
Rajendra Nagar and Shanmuganathan Raman.
\newblock Fast and accurate intrinsic symmetry detection.
\newblock In {\em Proceedings of the European Conference on Computer Vision
  (ECCV)}, pages 417--434, 2018.

\bibitem{ovsjanikov2012functional}
Maks Ovsjanikov, Mirela Ben-Chen, Justin Solomon, Adrian Butscher, and Leonidas
  Guibas.
\newblock Functional maps: a flexible representation of maps between shapes.
\newblock {\em ACM Transactions on Graphics (TOG)}, 31(4):30, 2012.

\bibitem{ovsjanikov2008global}
Maks Ovsjanikov, Jian Sun, and Leonidas Guibas.
\newblock Global intrinsic symmetries of shapes.
\newblock In {\em Computer graphics forum}, volume~27, pages 1341--1348. Wiley
  Online Library, 2008.

\bibitem{qi2017pointnet}
Charles~R Qi, Hao Su, Kaichun Mo, and Leonidas~J Guibas.
\newblock {PointNet}: Deep learning on point sets for 3d classification and
  segmentation.
\newblock {\em Proc. Computer Vision and Pattern Recognition (CVPR), IEEE},
  1(2):4, 2017.

\bibitem{Rodol2017deep}
Emanuele Rodol{\`{a}}, Zorah L{\"{a}}hner, Alexander~M. Bronstein, Michael~M.
  Bronstein, and Justin Solomon.
\newblock Functional maps representation on product manifolds.
\newblock {\em Comput. Graph. Forum}, 38(1):678--689, 2019.

\bibitem{rustamov2007laplace}
Raif~M Rustamov.
\newblock Laplace-beltrami eigenfunctions for deformation invariant shape
  representation.
\newblock In {\em Proceedings of the fifth Eurographics symposium on Geometry
  processing}, pages 225--233. Eurographics Association, 2007.

\bibitem{speciale2016symmetry}
Pablo Speciale, Martin~R Oswald, Andrea Cohen, and Marc Pollefeys.
\newblock A symmetry prior for convex variational 3d reconstruction.
\newblock In {\em European Conference on Computer Vision}, pages 313--328.
  Springer, 2016.

\bibitem{sumner2004deformation}
Robert~W Sumner and Jovan Popovi{\'c}.
\newblock Deformation transfer for triangle meshes.
\newblock {\em ACM Transactions on graphics (TOG)}, 23(3):399--405, 2004.

\bibitem{tevs2014relating}
Art Tevs, Qixing Huang, Michael Wand, Hans-Peter Seidel, and Leonidas Guibas.
\newblock Relating shapes via geometric symmetries and regularities.
\newblock {\em ACM Transactions on Graphics (TOG)}, 33(4):119, 2014.

\bibitem{vlasic2008articulated}
Daniel Vlasic, Ilya Baran, Wojciech Matusik, and Jovan Popovi{\'c}.
\newblock Articulated mesh animation from multi-view silhouettes.
\newblock 27(3):97, 2008.

\bibitem{wang2017group}
Hui Wang and Hui Huang.
\newblock Group representation of global intrinsic symmetries.
\newblock In {\em Computer Graphics Forum}, volume~36, pages 51--61. Wiley
  Online Library, 2017.

\bibitem{xu2012multi}
Kai Xu, Hao Zhang, Wei Jiang, Ramsay Dyer, Zhiquan Cheng, Ligang Liu, and
  Baoquan Chen.
\newblock Multi-scale partial intrinsic symmetry detection.
\newblock {\em ACM Transactions on Graphics (TOG)}, 31(6):181, 2012.

\bibitem{xu2009partial}
Kai Xu, Hao Zhang, Andrea Tagliasacchi, Ligang Liu, Guo Li, Min Meng, and
  Yueshan Xiong.
\newblock Partial intrinsic reflectional symmetry of 3d shapes.
\newblock {\em ACM Transactions on Graphics (TOG)}, 28(5):138, 2009.

\bibitem{yi2017syncspeccnn}
Li Yi, Hao Su, Xingwen Guo, and Leonidas~J Guibas.
\newblock {SyncSpecCNN}: Synchronized spectral {CNN} for 3d shape segmentation.
\newblock In {\em CVPR}, pages 6584--6592, 2017.

\end{thebibliography}
}

\end{document}